\documentclass[11pt]{article}
\usepackage[T1]{fontenc}
\usepackage[utf8]{inputenc}
\usepackage{lmodern}
\usepackage{microtype,fullpage,authblk}
\usepackage{amsmath,amssymb,amsthm,mathtools,bm}
\usepackage{mathpazo}
\usepackage{booktabs,xcolor,mdframed}
\usepackage{array,float,needspace,tikz}
\usetikzlibrary{arrows.meta}

\definecolor{MainBlue}{HTML}{173B57}
\definecolor{SoftBlue}{HTML}{F1F6FA}
\usepackage[numbers,sort&compress]{natbib}
\usepackage[colorlinks=true,linkcolor=blue!55!black,citecolor=blue!55!black,urlcolor=blue!55!black,pdfborder={0 0 0}]{hyperref}
\allowdisplaybreaks
\newtheorem{theorem}{Theorem}
\newtheorem{proposition}[theorem]{Proposition}
\newtheorem{lemma}[theorem]{Lemma}

\theoremstyle{remark}

\mdfdefinestyle{theoremframe}{
backgroundcolor=SoftBlue,
linecolor=MainBlue,
linewidth=.65pt,
innerleftmargin=7pt,innerrightmargin=7pt,
innertopmargin=6pt,innerbottommargin=6pt,
needspace=6\baselineskip,
nobreak=true,
skipabove=9pt,skipbelow=9pt
}
\surroundwithmdframed[style=theoremframe]{theorem}
\newcommand{\cN}{\mathcal N}

\newcommand{\cD}{\mathcal D}
\newcommand{\cL}{\mathcal L}
\newcommand{\cE}{\mathcal E}
\newcommand{\cA}{\mathcal A}
\newcommand{\cR}{\mathcal R}
\newcommand{\Tr}{\operatorname{Tr}}

\newcommand{\rank}{\operatorname{rank}}
\newcommand{\id}{\operatorname{id}}
\newcommand{\T}{\mathsf T}
\newcommand{\ket}[1]{\lvert#1\rangle}
\newcommand{\bra}[1]{\langle#1\rvert}
\newcommand{\proj}[1]{\ket{#1}\!\bra{#1}}
\numberwithin{equation}{section}

\title{\Large\textbf{Private communication via zero-private-capacity quantum channels}}
\date{September 9, 2026}

\author[1]{Chengkai Zhu}
\author[2,*]{Xin Wang}

\affil[1]{QudeLeap Research, Shanghai 200030, China}
\affil[2]{Thrust of Artificial Intelligence, Information Hub,\par
The Hong Kong University of Science and Technology (Guangzhou), Guangzhou 511453, China}

\hypersetup{pdftitle={Private communication via zero-private-capacity quantum channels},pdfauthor={Xin Wang and Chengkai Zhu}}

\begin{document}
\maketitle

\begingroup
\renewcommand{\thefootnote}{\fnsymbol{footnote}}

\footnotetext[1]
{\href{mailto:felixxinwang@hkust-gz.edu.cn}
{felixxinwang@hkust-gz.edu.cn}}

\endgroup

\begin{abstract}
Private communication over a noisy quantum channel requires reliable transmission to the receiver and secrecy from the environment. Whether two channels with zero private capacity can jointly enable private communication is a longstanding open problem in quantum information theory. Here we resolve this problem by exhibiting a four-level channel and a qubit erasure channel with half erasure probability, each with zero private capacity, whose joint use achieves more than 0.0001903 private bits per product use. The encoding gives the receiver a linear information gain with at most quadratic environmental leakage, enabling privacy through a fixed joint measurement and classical coding. This superactivation, impossible for independent classical memoryless wiretap channels, shows that a channel's private capacity alone does not determine its value for secure communication. The initial activation example was identified through interactions with large language models, and the result has been formalized in Lean 4.
\end{abstract}

\section{Introduction}
Private communication requires a message to reach the receiver reliably
while remaining hidden from the environment. Zero private classical
capacity excludes every positive asymptotic rate, even with arbitrary
mixed encodings and collective decoding over many uses
\cite{Wyner1975,CK1978,SW1998,CWY2004,Devetak2005}.
We construct two independent, finite-dimensional memoryless quantum
channels whose private capacities vanish separately but become positive
under joint use. Independent classical memoryless wiretap channels cannot
exhibit this phenomenon, as follows from the classical wiretap
capacity formula~\cite{CK1978}.

Quantum-capacity superactivation established an analogous phenomenon for
quantum transmission~\cite{SmithYard2008}, but one factor in that
construction already supports private communication. Smith and Smolin gave early conditional evidence that channels with
arbitrarily small individual private capacities could jointly support
a large quantum capacity~\cite{SmithSmolinEarly2009}. Known private-capacity
nonadditivity~\cite{LWZG2009,SmithSmolin2009} likewise does not supply two
zero-private-capacity factors. This stronger question was raised at
QIP 2009~\cite[slide~27]{SmithQIP2009} and linked to the gap between
information comparisons and physical channel simulation~\cite{HL2023}.

In this work, we resolve the private-capacity superactivation problem by finding two  finite-dimensional memoryless channels with zero individual private capacities and positive joint private capacity. First, we exhibit a four-level channel whose environment can reproduce all receiver measurement statistics at every blocklength, establishing zero private capacity. Second, we construct an explicit encoding and decoding scheme that achieves a positive private rate when this channel is paired with a qubit erasure channel of any fixed erasure probability $1/2\le p<1$, which also has zero private capacity (Theorem~\ref{thm:private}). The encoding adds a weak signal to a mixed background, producing a linear information gain for the receiver and at most quadratic environmental leakage in the signal strength. A fixed joint measurement on each output pair and classical coding across uses suffice; at half erasure, the achievable rate exceeds $1.9030\times10^{-4}$ private bits per product use. We further show that decoders whose effects are positive under partial transpose across the two output groups cannot achieve a positive private rate, even with collective processing within each group (Proposition~\ref{prop:ppt}). Thus zero private capacity is not preserved under tensor products, and joint quantum measurements can enable private communication that remains impossible under local operations and classical communication across the output groups.



\paragraph{Lean proof.}
The capacity inequalities in Theorem~\ref{thm:private}, interpreted
through Eq.~\eqref{eq:capacity}, and the fixed-measurement information
bound have been formalized in Lean~4 using Mathlib and
Lean-QIT~\cite{Lean4,Mathlib,LeanQIT}. Appendix~\ref{app:lean} describes
the scope and correspondence with the source code at
\url{https://github.com/QudeLeap/lean-private-capacity}.

\section{The channels and main result}
\label{sec:construction}
All Hilbert spaces are finite dimensional. The main input and output
are $A=B=\mathbb C^4$, with levels numbered $0,1,2,3$. Its noise consists
of a diagonal operator and three transitions together with their reverses:
\begin{equation}
\begin{aligned}
    D&=\operatorname{diag}(1,-4,4,-1),\\
    A_1&=\ket0\bra1-2\ket1\bra2+\ket2\bra3,\\
    A_2&=\ket0\bra2-\ket1\bra3,\qquad A_3=\ket0\bra3.
\end{aligned}
\label{eq:errors}
\end{equation}
For $X\in\mathcal L(\mathbb C^4)$, define $\cN:A\to B$ by
\begin{equation}
\boxed{\begin{aligned}
        \cN(X)=\frac1{56}\bigl[&8X+DXD
        +4(A_1XA_1^\dagger+A_1^\dagger XA_1)\\
        &+12(A_2XA_2^\dagger+A_2^\dagger XA_2)
        +31(A_3XA_3^\dagger+A_3^\dagger XA_3)\bigr].
\end{aligned}}
\label{eq:main-channel}
\end{equation}
This Kraus form is completely positive. It is trace preserving because
\begin{equation}
8I+D^2+\sum_{j=1}^3 c_j(A_j^\dagger A_j+A_jA_j^\dagger)=56I,
\qquad (c_1,c_2,c_3)=(4,12,31).
\label{eq:tp}
\end{equation}
Thus, an identity branch of weight $1/7$ is mixed with seven errors,
without revealing the branch to the receiver. The encoding in
Eq.~\eqref{eq:encoding} detects all seven nonidentity errors.
The weights also admit the transpose simulator constructed in
Appendix~\ref{app:simulator}.

The helper is a qubit erasure channel~\cite{BDS1997} from $R$ to $R_B=\mathbb C^2\oplus\mathbb C\ket e$ and $0\le p\le1$:
\begin{equation}
\boxed{
\cE_{2,p}(X)=(1-p)X\oplus p\Tr(X)\proj e,
\qquad \ket e\perp\mathbb C^2.
\label{eq:erasure}
}
\end{equation}
It reaches Bob with probability $1-p$ and Eve with probability $p$. Orthogonal flags record these alternatives.

All logarithms are base two unless written $\ln$. Let $S$ be von Neumann
entropy, $h_2(x)=-x\log x-(1-x)\log(1-x)$ binary entropy
(with $0\log0=0$), and $I$ classical mutual information.
An ensemble $\{p_u,\omega_u\}$ has Holevo information
$\chi(U:B)=S(\sum_up_u\omega_u)-\sum_up_uS(\omega_u)$.
For a quantum channel $\Phi:A\to B$, fix a complementary channel
$\Phi^c:A\to E$ obtained from a Stinespring isometry. The private capacity of $\Phi$ is given by~\cite{CWY2004,Devetak2005}
\begin{equation}
P(\Phi)=\sup_{n\ge1}\frac1n\sup_{\{p_u,\rho_u\}}
\bigl[\chi(U:B^n)-\chi(U:E^n)\bigr],
\label{eq:capacity}
\end{equation}
where the inner supremum ranges over finite ensembles of density
operators on $A^{\otimes n}$. The two Holevo quantities are evaluated
on $\{p_u,\Phi^{\otimes n}(\rho_u)\}$ and
$\{p_u,(\Phi^c)^{\otimes n}(\rho_u)\}$, respectively.

Precisely, an $(n,M)$ code uses $M$ equiprobable mixed inputs $\rho_m$
and a receiver POVM $\{F_m\}$. Set $\beta_m=\Phi^{\otimes n}(\rho_m)$,
$\eta_m=(\Phi^c)^{\otimes n}(\rho_m)$, and
$\bar\eta=M^{-1}\sum_m\eta_m$. Its errors are
\begin{equation}
\varepsilon=1-\frac1M\sum_m\Tr(F_m\beta_m),\qquad
\delta=\frac1{2M}\sum_m\|\eta_m-\bar\eta\|_1.
\label{eq:errors-operational}
\end{equation}
A rate $r$ is achievable if codes have $\liminf_{n\to\infty}(\log M_n)/n\ge r$
and $\varepsilon_n,\delta_n\to0$. No preshared key, entanglement, or
auxiliary communication is supplied; encoder randomness is local and
private. One product use consumes one use of each factor.
Appendix~\ref{app:operational} relates this trace-distance convention
to Eq.~\eqref{eq:capacity}.

\begin{theorem}[Private-capacity superactivation]
\label{thm:private}
The channels in Eqs.~\eqref{eq:main-channel} and \eqref{eq:erasure}
satisfy $P(\cN)=P(\cE_{2,p})=0$ for $1/2\le p\le1$.
For every fixed $1/2\le p<1$,
\begin{equation}
    P(\cE_{2,p}\otimes\cN)\ge\frac{6\ln2}{49(27+169p)}(1-p)^2>0.
    \label{eq:rate}
\end{equation}
Every rate below this bound is achievable with a fixed binary measurement
on each product output and classical wiretap coding across uses.
\end{theorem}
The zero-capacity statements allow arbitrary block encodings and
collective decoding. Each environment has at least as much Holevo information as its receiver for every input ensemble and every
blocklength. The product channel admits a positive private rate using
a fixed measurement within each use and classical coding across uses.
Section~\ref{sec:meaning} shows why PPT decoding cannot achieve such a rate.
At half erasure, Eq.~\eqref{eq:rate} gives
$3\ln2/10927\simeq1.903\times10^{-4}$ private bits per product use.

\section{Proof of Theorem~\ref{thm:private}}
\label{sec:proof}
We first prove the two zero-capacity statements. We then compare Bob's
information from a fixed binary measurement with Eve's leakage for the
same input ensemble.

\subsection{Why each factor has zero private capacity}
Let $\cN^c:A\to E$ be the full complement fixed in
Appendix~\ref{app:simulator}. That appendix constructs a completely
positive trace-preserving (CPTP) map $\cD:E\to B$ such that
\begin{equation}
\cD\cN^c=\T_B\cN,
\label{eq:transpose}
\end{equation}
where $\T_B$ is linear matrix transposition in the numbered basis.
This is a transpose-antidegradability certificate, in the conjugate
simulation framework~\cite{BradlerEtAl2010}. At blocklength $n$,
$\cD^{\otimes n}$ reproduces the global transpose of Bob's output.
For every finite ensemble on $A^{\otimes n}$,
\[
\chi\bigl(\{p_u,\cN^{\otimes n}(\rho_u)\}\bigr)
=\chi\bigl(\{p_u,\cD^{\otimes n}(\cN^c)^{\otimes n}(\rho_u)\}\bigr)
\le\chi\bigl(\{p_u,(\cN^c)^{\otimes n}(\rho_u)\}\bigr).
\]
The equality uses entropy invariance under global transposition; the
inequality uses data processing for the CPTP map $\cD^{\otimes n}$.
Thus $P(\cN)=0$, without restricting block inputs to product states.
For $p\ge1/2$, the erasure environment simulates Bob by retaining its
received qubit with probability $(1-p)/p$ and erasing otherwise.
Thus $P(\cE_{2,p})=0$ as well.

\subsection{A signal event absent from the background}
We now use both channels. The encoding has three orthogonal directions:
two form the background, and the third carries the signal. On $RA$, choose
\begin{equation}
u_0=\frac{2\ket{00}+\ket{11}}{\sqrt5},\qquad
u_1=\frac{\ket{02}+2\ket{13}}{\sqrt5},\qquad
w=\frac{\ket{01}+\ket{12}}{\sqrt2}.
\label{eq:encoding}
\end{equation}
These vectors are orthonormal. Define the background and signal states
\begin{equation}
\rho_0=\tfrac12(\proj{u_0}+\proj{u_1}),\qquad \rho_1=\proj w.
\label{eq:inputs}
\end{equation}
Alice's two actual letters are $\rho_0$ and
$\rho_t=(1-t)\rho_0+t\rho_1$, used with probabilities $3/4$ and $1/4$.
Here $0<t\le1/2$ will be chosen below. The mixtures use Alice's local
randomness. Both letters have helper marginal $I_R/2$, so the helper
alone carries no letter information.

The amplitudes in~\eqref{eq:encoding} make all seven errors detectable:
each error sends the encoding space into its orthogonal complement.
Specifically, let $P_{\mathcal S}$ project onto the span of $u_0,u_1,w$.
For every
nonidentity noise operator $L\in\{D,A_1,A_1^\dagger,A_2,A_2^\dagger,A_3,A_3^\dagger\}$,
\begin{equation}
P_{\mathcal S}(I_R\otimes L)P_{\mathcal S}=0.
\label{eq:detection}
\end{equation}
For example, the relevant diagonal and transition amplitudes cancel as
$\langle u_0|I\otimes D|u_0\rangle=(4-4)/5=0$ and
$\langle u_0|I\otimes A_1|w\rangle=(2-2)/\sqrt{10}=0$.
The other entries vanish by the same cancellations or disjoint support.
Consequently, for every density operator $\rho=P_{\mathcal S}\rho P_{\mathcal S}$,
\[
(\id_R\otimes\cN)(\rho)=\rho/7+\beta_\perp,\qquad
\beta_\perp\succeq0,\quad P_{\mathcal S}\beta_\perp=0.
\]
Thus, only the identity branch can produce an outcome in $\mathcal S$.
Within that branch, the background has no component along $w$.

When the helper arrives, Bob measures $\{\proj w,I-\proj w\}$; otherwise, he records no click. On
$R_BB=(\mathbb C^2\otimes B)\oplus(\mathbb C\ket e\otimes B)$ the
full effects are $F_1=\proj w\oplus0_B$ and $F_0=I-F_1$.
For his binary outcome $Y$ and Alice's
letter $U\in\{0,1\}$,
\begin{equation}
\Pr(Y=1\mid U=0)=0,\qquad
\Pr(Y=1\mid U=1)=a_pt,\qquad a_p=\frac{1-p}{7}.
\label{eq:click}
\end{equation}
A click identifies $U=1$ with certainty. Its contribution to
mutual information is $(a_pt/4)\log_2 4=a_pt/2$; the no-click contribution
is nonnegative because it is a weighted relative entropy. Consequently,
\begin{equation}
I(U:Y)=h_2(a_pt/4)-\tfrac14h_2(a_pt)\ge\tfrac12a_pt.
\label{eq:bob-gain}
\end{equation}

\subsection{Bounding Eve's leakage by a classical coin}
When the helper is erased at Bob, it is available to Eve. Let $E'$
denote the helper environment. For the background and signal states,
define
\begin{equation}
\sigma_i=(\id_R\otimes\cN^c)(\rho_i),\qquad
\epsilon_i=\Tr_R\sigma_i,\qquad
\gamma_i^{(p)}=p\sigma_i\oplus(1-p)\epsilon_i.
\label{eq:environment}
\end{equation}
Here, $i=0,1$ indexes the background and signal component. For the actual
second letter Eve receives $\gamma_t^{(p)}=(1-t)\gamma_0^{(p)}+t\gamma_1^{(p)}$.
The direct sum identifies the full environment $E'E$ in the erasure dilation specified in Appendix~\ref{app:environment}: Eve receives the helper in the first block and does not in the second.
Appendix~\ref{app:environment} proves
\begin{equation}
\sigma_1\preceq\frac{205}{9}\sigma_0,
\qquad \epsilon_1\preceq4\epsilon_0.
\label{eq:order}
\end{equation}
The only nontrivial block comparison is a $2\times2$ rank-one identity.
These bounds establish support containment even when Eve receives $R$.

We bound Eve's Holevo information by expressing her conditional states
as outputs of a fixed preparation channel driven by a biased
classical coin.
\begin{lemma}[Classical coin bound]
\label{lem:coin}
Let $\omega_0,\omega_1$ be density operators on the same finite-dimensional
space, with $\omega_1\preceq c\omega_0$ for a finite $c\ge3$.
For $0\le t\le1/2$, the ensemble with letters $\omega_0$ and
$\omega_t=(1-t)\omega_0+t\omega_1$ and probabilities $3/4,1/4$
has Holevo information at most $3(c-1)t^2/(32\ln2)$.
\end{lemma}
\begin{proof}
For $0\le s\le1/2$, set $\omega_s=(1-s)\omega_0+s\omega_1$ and write
\begin{equation}
    \omega_s=r_s\omega_1+(1-r_s)\tau,\qquad
    r_s=\frac{1+(c-1)s}{c},\qquad
    \tau=\frac{c\omega_0-\omega_1}{c-1}.
    \label{eq:coin}
\end{equation}
The domination assumption gives $\tau\succeq0$, and $\Tr\tau=1$.
Also $0<r_s<1$ on the stated interval. A coin of bias $r_s$ followed by
preparation of $\omega_1$ or $\tau$ therefore produces $\omega_s$.
Data processing bounds Eve's Holevo information by the mutual information
of the coin. This construction does not require the states to commute.

Both coin outcomes have nonzero probability at $s=0$. We bound the
entropy curvature by setting $f(s)=h_2(r_s)$. On $0\le s\le1/2$,
\begin{equation}
    -f''(s)=\frac{c-1}{\ln2\,[1+(c-1)s](1-s)}
    \le\frac{c-1}{\ln2},
    \label{eq:curvature}
\end{equation}
since $[1+(c-1)s](1-s)=1+s[(c-2)-(c-1)s]\ge1$ for $c\ge3$.
Convexity of $f(s)+(c-1)s^2/(2\ln2)$ therefore gives
\begin{equation}
    \chi(U:E)\le f(t/4)-\tfrac34 f(0)-\tfrac14 f(t)
    \le\frac{3(c-1)t^2}{32\ln2}.
    \label{eq:coin-bound}
\end{equation}
\end{proof}
The erasure flag is independent of $U$, so Eve's Holevo information is
the average of the two branch values. Applying~\eqref{eq:coin-bound}
to~\eqref{eq:order} uses the direct-sum entropy identity
$S(p\sigma\oplus(1-p)\epsilon)=h_2(p)+pS(\sigma)+(1-p)S(\epsilon)$
and yields
\begin{equation}
\chi(U:E'E)\le\frac{3\kappa_p t^2}{32\ln2},\qquad
\kappa_p=p\left(\frac{205}{9}-1\right)+(1-p)(4-1)
=\frac{27+169p}{9},
\label{eq:eve-cost}
\end{equation}
This quadratic bound includes both erasure branches and the full
quantum environment.

\subsection{From an information advantage to a private rate}
Combining~\eqref{eq:bob-gain} and~\eqref{eq:eve-cost},
\begin{equation}
I(U:Y)-\chi(U:E'E)\ge\frac{a_pt}{2}
-\frac{3\kappa_p t^2}{32\ln2}.
\label{eq:gain}
\end{equation}
For every fixed $p<1$, $a_p>0$, so a sufficiently weak signal gives a
positive information difference. Choose $t=8a_p\ln2/(3\kappa_p)$. For $1/2\le p<1$,
$a_p\le1/14$ and $\kappa_p\ge12$, so
$0<t\le\ln2/63<1/2$. Substitution gives Eq.~\eqref{eq:rate}.
To obtain an operational rate, repeat this encoding and fixed measurement.
They define a memoryless wiretap channel with classical input $U$,
classical receiver output $Y$, and Eve's unchanged quantum marginal.
The private coding theorem~\cite{Devetak2005} achieves every rate below
$I(U:Y)-\chi(U:E'E)$ by classical coding over these uses.
Appendix~\ref{app:operational} gives the induced channel explicitly and
checks the trace-distance secrecy criterion. Fixing $p$ and $t$ before
the coding limit proves Theorem~\ref{thm:private}.

\section{Why the measurement must be joint}
\label{sec:meaning}
For the channel $\cN$ alone, Eve can measure the transpose of Bob's
POVM on her simulated transposed state and reproduce his outcome
distribution. For its product with an antidegradable channel, the
simulation instead transposes only one output factor. This still works for every decoder whose effects remain
positive under that partial transpose, giving the following converse.

\begin{proposition}[A finite-code bound for PPT decoding]
\label{prop:ppt}
Let $\cN:A_1\to B_1$ and $\cA:A_2\to B_2$ be finite-dimensional
channels with full complements on $E_1,E_2$. Suppose $\cN$ admits a CPTP simulator $\cD\cN^c=\T_{B_1}\cN$ and
$\cA$ is antidegradable. Write $E=E_1E_2$. For any blocklength $n\ge1$, consider a code
for $(\cN\otimes\cA)^{\otimes n}$ with $M\ge2$ uniform messages,
arbitrary mixed inputs, and decoder $\{F_m\}$ on $B_1^nB_2^n$.
If $F_m^{\T_{B_1^n}}\succeq0$ for every $m$, then
\begin{equation}
    \varepsilon+\delta\ge1-1/M,
    \label{eq:ppt-bound}
\end{equation}
where $\varepsilon$ is average decoding error and
$\delta=\tfrac12\|\rho_{UE^n}-\pi_U\otimes\rho_{E^n}\|_1$ is secrecy
distance from the uniform message independent of the full environment,
where $\pi_U=I_U/M$ and $\rho_{UE^n}=M^{-1}\sum_m\proj m\otimes\eta_m$. Fixed regrouping of the output factors is understood.
\end{proposition}
\begin{proof}
Choose a CPTP simulator $\cR\cA^c=\cA$ and set $\cL=\cD\otimes\cR$.
For each message $x$, Eve's and Bob's outputs satisfy
$\cL^{\otimes n}(\eta_x)=\beta_x^{\T_{B_1^n}}$.
The operators $G_m=(\cL^\dagger)^{\otimes n}(F_m^{\T_{B_1^n}})$
are positive and sum to the identity, hence form an Eve POVM. They obey
\[
\Tr(G_m\eta_x)=
\Tr(F_m^{\T_{B_1^n}}\beta_x^{\T_{B_1^n}})=\Tr(F_m\beta_x).
\]
The joint guessing test $\sum_m\proj m_U\otimes G_m$ is an effect.
It succeeds with probability $1-\varepsilon$ on $\rho_{UE^n}$ and $1/M$
on $\pi_U\otimes\rho_{E^n}$. The variational characterization of trace
distance therefore gives $1-\varepsilon-1/M\le\delta$.
\end{proof}

Positive-partial-transpose (PPT) decoding thus has zero private rate.
This includes separable measurements and those implemented by local
operations and classical communication across the two factors
\cite{MatthewsEtAl2009}. Our effect $\proj w$ has partial-transpose
eigenvalue $-1/2$. The protocol uses this measurement only within each
product use; classical coding suffices across uses. The distinction concerns
the measurement: because $\T\cN=\cD\cN^c$ is completely positive,
Bob's output is itself PPT across the factors even for entangled inputs. PPT positivity alone does not preclude privacy, e.g., PPT entangled states
can have positive distillable key~\cite{HHHO2005}.

\section{Discussion}
The background mixture and signal dilution perform different tasks:
the first fills Eve's relevant support while leaving Bob's signal
direction empty; the second makes leakage smaller than the receiver's
gain. More generally, the same proof applies whenever a receiver event has zero background probability and positive signal probability, while the environmental signal support is contained in the background support. On the background support, the background state is invertible, so
$c=\max\{3,\|\omega_0^{-1/2}\omega_1\omega_0^{-1/2}\|_\infty\}$
is finite and satisfies $\omega_1\preceq c\omega_0$. The inverse is restricted to that support. Related support arguments detect positive quantum capacity
\cite{SinghDatta2022}. Here, they certify privacy for mixed letters under a fixed receiver measurement.

The relevant information comparison is the regularized less-noisy
order~\cite{Watanabe2012,HircheRouzeFranca2022}: the complement
$\cN^c$ dominates $\cN$ in Holevo information at every blocklength.
However, $\cN^c$ is not degradable, since otherwise $\cN$ would be
antidegradable and could not activate with the erasure channel.
The construction therefore, gives an explicit separation between
degradability and the regularized less-noisy order, through the
connection to private-capacity superactivation established in~\cite{HL2023}. A qutrit channel with zero private capacity, despite being neither
PPT, nor antidegradable, was constructed in
Ref.~\cite{ZhuWang2026Incapacity}. Its complement dominates the receiver even in the presence of an arbitrary quantum reference.
This complete less-noisy order tensorizes~\cite{HircheRouzeFranca2022}
and therefore excludes activation with an antidegradable helper.
Our example shows that information domination at every blocklength alone does not provide this stability.


For the outlook, the remaining structural question is which zero-private-capacity channels admit the support asymmetry exhibited here and hence can turn a failure of joint simulation into an explicit private protocol. It would also be important to understand whether the classical capacity of quantum channels can exhibit superadditivity.

\paragraph{Acknowledgements.} The authors thank Andreas Winter for helpful discussions. XW was partially supported by the National Natural Science Foundation of China (Grant No.~92576114), the Guangdong Provincial Quantum Science Strategic Initiative (Grant No.~GDZX2403008, GDZX2503001), and the Guangdong Provincial Key Lab of Integrated Communication, Sensing and Computation for Ubiquitous Internet of Things (Grant No.~2023B1212010007).

\section*{AI statement}
This work was developed with assistance from QudeLeap's AI Quantum Scientist, a system under development. An initial activating example was identified and the proof strategies was explored through interactions with large language models. The authors checked, rewrote and organized the proofs, and prepared the manuscript. 
Artificial intelligence was also
used to prepare the initial manuscript and helped with revisions.
The authors retain  responsibility for the correctness, attribution, and presentation of the final results.

\appendix
\section{An explicit transpose simulator}
\label{app:simulator}
Fix the eight Kraus operators $K_e=\sqrt{w_e}L_e$, where
\begin{equation}
\begin{aligned}
    (L_0,\ldots,L_7)&=(I,D,A_1,A_1^\dagger,A_2,A_2^\dagger,A_3,A_3^\dagger),\\
    (w_0,\ldots,w_7)&=\tfrac1{56}(8,1,4,4,12,12,31,31).
\end{aligned}
\label{eq:kraus}
\end{equation}
The isometry $V_{\cN}:A\to B\otimes E$ is
$V_{\cN}\ket\psi=\sum_e K_e\ket\psi\otimes\ket e$; Kraus completeness
gives $V_{\cN}^\dagger V_{\cN}=I_A$. Its full complement is
$\cN^c(X)=\Tr_B(V_{\cN}XV_{\cN}^\dagger)$, with matrix elements
$[\cN^c(X)]_{ef}=\Tr(K_eXK_f^\dagger)$.
We construct a CPTP map satisfying Eq.~\eqref{eq:transpose}.

Identify the simulator output $C=\mathbb C^4$ with $B$ in the numbered
bases. With $E=\mathbb C^8$ and
$\ket{e,c}=\ket e_E\otimes\ket c_C$, define
\begin{equation}
\begin{aligned}
    v_1&=\tfrac12\ket{0,0}-\tfrac18\ket{1,0}+\ket{3,1},\\
    v_2&=\tfrac12\ket{0,0}+\tfrac18\ket{1,0}+\ket{5,2},\qquad
    v_3=\ket{7,3},\\
    v&=2\ket{0,2}+\ket{1,2}-2\ket{2,1}+\ket{3,3}+3\ket{4,0},\\
    a&=\ket{3,0}-2\ket{5,1}+\ket{7,2},\qquad
    b=-\ket{5,0}+\ket{7,1},\qquad z=\ket{7,0}.
\end{aligned}
\label{eq:sim-vectors}
\end{equation}
Set $V=(v_1,v_2,v_3):\mathbb C^3\to E\otimes C$ and
\begin{equation}
M=\begin{pmatrix}16&8&29\\8&64&77\\29&77&251\end{pmatrix},\qquad
Q=VMV^\dagger+4(\proj v+\proj a+\proj b)+51\proj z.
\label{eq:sim-positive}
\end{equation}
The leading principal minors of $M$ are $16$, $960$, and $128000$;
therefore $M\succ0$ and $Q\succeq0$.

Define the real orthogonal reflection
\begin{equation}
\mathsf R\ket{e,c}=s_e\ket{\pi(e),3-c},\qquad
\begin{aligned}
    \pi&=(0,1,3,2,5,4,7,6),\\
    s&=(1,-1,1,1,-1,-1,1,1),
\end{aligned}
\label{eq:reflection}
\end{equation}
where the entries list the images and signs in order $e=0,\ldots,7$.
With $W=\operatorname{diag}(\sqrt{w_0},\ldots,\sqrt{w_7})$, set
\begin{equation}
H=\frac{Q+\mathsf RQ\mathsf R^\dagger}{560},\qquad
J_{\cD}=(W^{-1}\otimes I_C)H(W^{-1}\otimes I_C).
\label{eq:sim-choi}
\end{equation}
All $w_e$ in~\eqref{eq:kraus} are positive. We use the unnormalized,
input-first Choi convention
$J_{\cD}=\sum_{e,f}\ket e\bra f\otimes\cD(\ket e\bra f)$.
Thus, writing $H_{ef}$ for the $4\times4$ blocks of $H$, the map is explicitly
\begin{equation}
\cD(Z)=\sum_{e,f=0}^7\frac{Z_{ef}}{\sqrt{w_ew_f}}H_{ef}.
\label{eq:sim-map}
\end{equation}
Positivity of $J_{\cD}$ follows from~\eqref{eq:sim-positive} by congruence.

Trace preservation and the composition identity are the following
finite rational equalities:
\begin{equation}
\Tr_C H=W^2,\qquad
\sum_{e,f=0}^7(L_f^\dagger L_e)_{ji}H_{ef}
=\cN(\ket i\bra j)^\T\quad(0\le i,j\le3).
\label{eq:sim-identities}
\end{equation}
These identities follow by expanding
Eqs.~\eqref{eq:sim-vectors}--\eqref{eq:sim-choi}. In the second identity the coefficient
$(L_f^\dagger L_e)_{ji}$ follows from
$\Tr(L_e\ket i\bra jL_f^\dagger)=(L_f^\dagger L_e)_{ji}$.
The first equality in~\eqref{eq:sim-identities} gives
$\Tr_CJ_{\cD}=I_E$. For the second, the factors $\sqrt{w_ew_f}$
from $[\cN^c(\ket i\bra j)]_{ef}$ cancel those in~\eqref{eq:sim-map}.
Thus $\cD$ is CPTP and $\cD\cN^c=\T\cN$, as required.

\section{Environmental support and order}
\label{app:environment}
Fix the erasure isometry $J_p:R\to R_B\otimes E'$, where both outputs
are $\mathbb C^2\oplus\mathbb C\ket e$, by
\[
J_p\ket\psi=\sqrt{1-p}\ket\psi_{R_B}\ket e_{E'}
+\sqrt p\ket e_{R_B}\ket\psi_{E'}.
\]
The two terms occupy orthogonal sectors in each output. Tracing out $R_B$
therefore removes their cross terms even when $R$ is correlated with $A$,
which proves the full-environment direct sum in~\eqref{eq:environment}.
The product dilation $V_p$ is $J_p\otimes V_{\cN}$ with outputs regrouped
as $R_BB:E'E$.

We retain both the helper and all eight environment levels. Number the
basis of $R\otimes E$ by $\ket j=\ket{r,e}$ with $j=8r+e$.
Use $W$ from Appendix~\ref{app:simulator} and remove only the positive
Kraus weights by the algebraic congruence
\begin{equation}
\widetilde\sigma_i=(I_R\otimes W^{-1})\sigma_i(I_R\otimes W^{-1}).
\label{eq:env-congruence}
\end{equation}
These unweighted matrices need not be normalized states. Let $T$ have
the following eight linearly independent columns:
\begin{equation}
T=\bigl(\ket0,\ \ket1+\tfrac12\ket{10},\
\ket8,\ \ket9-\tfrac12\ket3,\
\ket2+\ket{12},\ \ket5-\ket{11},\
\ket4+2\ket{14},\ 2\ket7-\ket{13}\bigr).
\label{eq:env-basis}
\end{equation}
Set
\begin{equation}
A=\begin{pmatrix}1/2&4/5\\4/5&2\end{pmatrix},\qquad
B=\begin{pmatrix}1/2&-2\\-2&8\end{pmatrix},\qquad Z=\operatorname{diag}(1,-1).
\label{eq:env-small}
\end{equation}
Substituting the Kraus operators in~\eqref{eq:kraus} into the
complementary-channel formula gives the full sixteen-dimensional
identities
\begin{equation}
\begin{aligned}
    \widetilde\sigma_0&=T\bigl(A\oplus ZAZ\oplus\tfrac25 I_2
    \oplus\tfrac1{10}I_2\bigr)T^\dagger,\\
    \widetilde\sigma_1&=T\bigl(B\oplus ZBZ\oplus\tfrac12 I_2
    \oplus0_2\bigr)T^\dagger.
\end{aligned}
\label{eq:env-blocks}
\end{equation}
They include all components and cross terms; $T$ is not assumed isometric.
Since $A_{00}=1/2>0$ and $\det A=9/25>0$, $A\succ0$.
The full column rank of $T$ and invertibility of $I_R\otimes W$ show that
$\sigma_0$ is positive on $\operatorname{ran}[(I_R\otimes W)T]$, and
both environmental outputs vanish on its orthogonal complement. The nontrivial order follows from
\begin{equation}
\boxed{\frac{205}{9}A-B=\frac29
    \begin{pmatrix}7\\13\end{pmatrix}
    \begin{pmatrix}7&13\end{pmatrix}\succeq0.}
\label{eq:env-rank-one}
\end{equation}
The scalar signal-to-background ratios are $5/4$ and zero, so
congruence proves $\sigma_1\preceq(205/9)\sigma_0$. In particular,
$\rank\sigma_0=8$ and $\rank\sigma_1=4$, with support containment.

For the reduced environment, direct partial trace gives
\begin{equation}
\begin{aligned}
    W^{-1}\epsilon_0W^{-1}
    &=\operatorname{diag}(1,4,9/10,9/10,1/2,1/2,2/5,2/5),\\
    W^{-1}\epsilon_1W^{-1}
    &=\operatorname{diag}(1,16,5/2,5/2,1/2,1/2,0,0).
\end{aligned}
\label{eq:reduced-env}
\end{equation}
The coordinate ratios are $1,4,25/9,25/9,1,1,0,0$, proving
$\epsilon_1\preceq4\epsilon_0$. This sharper reduced-environment bound
is why the coin cost is averaged over the two erasure branches.

\section{Coding and the operational secrecy convention}
\label{app:operational}
\subsection{The fixed measurement preserves Eve's marginal}
Let $V_p:RA\to R_BB E'E$ be the product
Stinespring isometry and put $\widehat\rho_0=\rho_0$,
$\widehat\rho_1=\rho_t$. The fixed measurement defines the memoryless
wiretap channel with classical input and receiver output
\begin{equation}
u\longmapsto\sum_{y=0}^1\proj y_Y\otimes
\Tr_{R_BB}\!\left[(\sqrt{F_y}\otimes I)V_p\widehat\rho_u V_p^\dagger
(\sqrt{F_y}\otimes I)\right].
\label{eq:measured-wiretap}
\end{equation}
Summing over $y$ leaves Eve's original marginal unchanged; Bob's outcome
is not publicly disclosed. The direct coding theorem~\cite[Theorem~1]{Devetak2005}
applies without auxiliary communication and achieves every rate below
$I(U:Y)-\chi(U:E'E)$. Its vanishing Holevo leakage implies the secrecy
criterion~\eqref{eq:errors-operational} by quantum Pinsker:
$\delta_n\le\sqrt{(\ln2)\chi(\mathsf M:E'^nE^n)/2}$.
Here, $\mathsf M$ is the uniform transmitted message register. Since Bob's output is classical,
any block decoder can be implemented as classical processing of $Y^n$.
Throughout this coding argument, $p$ and $t$ are fixed before $n\to\infty$.

\subsection{The trace-distance capacity formula}
Write $a_n$ for the inner supremum in Eq.~\eqref{eq:capacity} and
$d_E=\dim E$. For any $(n,M)$ code with $M\ge2$, let
$\rho_{UE^n}=M^{-1}\sum_m\proj m\otimes\eta_m$. Its secrecy error is
$\delta=\tfrac12\|\rho_{UE^n}-\pi_U\otimes\bar\eta\|_1$.
Entropy continuity~\cite{Audenaert2007}, applied in dimension $Md_E^n$, gives
\[
\chi(U:E^n)=S(\pi_U\otimes\bar\eta)-S(\rho_{UE^n})
\le\delta(\log M+n\log d_E)+h_2(\delta).
\]
Fano's inequality and data processing give
$\chi(U:B^n)\ge(1-\varepsilon)\log M-h_2(\varepsilon)$.
Since the code ensemble is allowed in the definition of $a_n$,
\begin{equation}
(1-\varepsilon-\delta)\log M
\le a_n+\delta n\log d_E+h_2(\varepsilon)+h_2(\delta).
\label{eq:operational-converse}
\end{equation}
Dividing by $n$ and taking $\varepsilon,\delta\to0$ proves the converse
in Eq.~\eqref{eq:capacity}. For achievability, fix a block size and a
finite ensemble, apply the private coding theorem~\cite{Devetak2005}
to repeated blocks, and use quantum Pinsker to convert vanishing Holevo
leakage into vanishing $\delta$. Optimizing over block sizes and ensembles
establishes the stated capacity formula.

\clearpage
\section{Lean correspondence}
\label{app:lean}

The formalization covers both zero-capacity statements for
$1/2\le p\le1$ and the positive lower bound~\eqref{eq:rate} for
$1/2\le p<1$. Capacity is defined by Eq.~\eqref{eq:capacity}, with
arbitrary finite mixed-state ensembles at every positive blocklength.
The main declaration is \texttt{superactivation\_main} in namespace
\texttt{QIT.Transition}. Its only hypotheses are the stated bounds on $p$. The channels, full
complements, input states, and receiver POVM are constructed explicitly.

The proof uses Lean~4.30.0, Mathlib, and foundational modules from
Lean-QIT~\cite{Lean4,Mathlib,LeanQIT}. Its transitive axiom dependencies
are \texttt{propext}, \texttt{Classical.choice}, and \texttt{Quot.sound};
there are no admitted steps or additional axioms in the checked proof.
Source code, pinned dependencies, and reproduction instructions are at
\url{https://github.com/QudeLeap/lean-private-capacity}.

Figure~\ref{fig:lean-dependencies} summarizes the argument, and
Table~\ref{tab:lean-correspondence} identifies the main declarations.
Write $r_p$ for the lower bound in Eq.~\eqref{eq:rate} and
$R_{\mathrm{meas}}=I(U:Y)-\chi(U:E'E)$ at the signal weight chosen
in Sec.~\ref{sec:proof}.

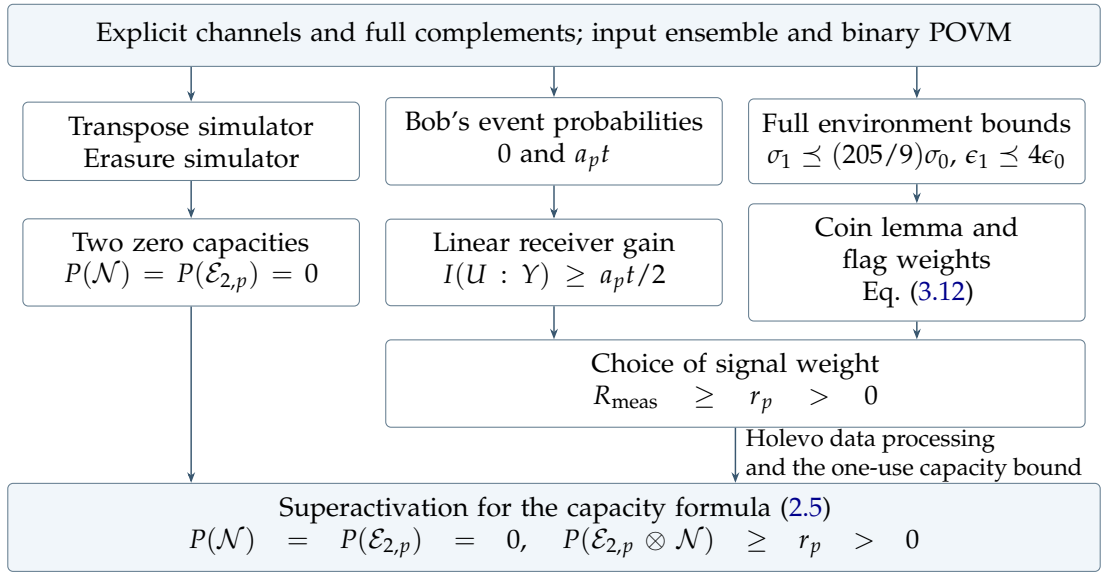
\begin{figure}[H]
\centering
\begin{tikzpicture}[
item/.style={draw=MainBlue!65,rounded corners=2pt,align=center,
font=\small,inner sep=5pt,minimum height=.82cm},
dep/.style={-{Stealth[length=4pt]},draw=MainBlue!85,line width=.55pt},
x=1cm,y=1cm]
\node[item,fill=SoftBlue,text width=14.1cm] (objects) at (0,0)
{Explicit channels and full complements; input ensemble and binary POVM};
\node[item,text width=4.1cm] (sim) at (-4.8,-1.4)
{Transpose simulator\\Erasure simulator};
\node[item,text width=4.1cm] (bob) at (0,-1.4)
{Bob's event probabilities\\$0$ and $a_pt$};
\node[item,text width=4.1cm] (order) at (4.8,-1.4)
{Full environment bounds\\$\sigma_1\preceq(205/9)\sigma_0$, $\epsilon_1\preceq4\epsilon_0$};
\node[item,text width=4.1cm] (zero) at (-4.8,-3)
{Two zero capacities\\$P(\cN)=P(\cE_{2,p})=0$};
\node[item,text width=4.1cm] (gain) at (0,-3)
{Linear receiver gain\\$I(U:Y)\ge a_pt/2$};
\node[item,text width=4.1cm] (coin) at (4.8,-3)
{Coin lemma and flag weights\\Eq.~\eqref{eq:eve-cost}};
\node[item,text width=8.9cm] (rate) at (2.4,-4.6)
{Choice of signal weight\\$R_{\mathrm{meas}}\ge r_p>0$};
\node[item,fill=SoftBlue,text width=14.1cm] (main) at (0,-6.5)
{Superactivation for the capacity formula~\eqref{eq:capacity}\\
$P(\cN)=P(\cE_{2,p})=0$, \quad $P(\cE_{2,p}\otimes\cN)\ge r_p>0$};
\draw[dep] ([xshift=-4.8cm]objects.south) -- (sim.north);
\draw[dep] (objects.south) -- (bob.north);
\draw[dep] ([xshift=4.8cm]objects.south) -- (order.north);
\draw[dep] (sim) -- (zero);
\draw[dep] (bob) -- (gain);
\draw[dep] (order) -- (coin);
\draw[dep] (gain.south) -- ([xshift=-2.4cm]rate.north);
\draw[dep] (coin.south) -- ([xshift=2.4cm]rate.north);
\draw[dep] (zero.south) -- ([xshift=-4.8cm]main.north);
\draw[dep] (rate.south) -- node[right,font=\footnotesize,align=left]
{Holevo data processing\\and the one-use capacity bound}
([xshift=2.4cm]main.north);
\end{tikzpicture}
\caption{Logical dependencies of the formalized argument for
$1/2\le p<1$. Arrows run from proof ingredients to conclusions, with
each box grouping the relevant Lean declarations. The final capacity statement uses the regularized information formula.}
\label{fig:lean-dependencies}
\end{figure}

Positivity is established by Gram or LDL decompositions. For the flagged
leakage bound, Lean constructs one preparation channel for the two
environmental branches. These proofs establish the displayed
inequalities while using intermediate representations adapted to finite matrix calculations and the available entropy library.

The operational coding theorem and its trace-distance formulation in
Appendix~\ref{app:operational}, Proposition~\ref{prop:ppt}, and the
further consequences in the Discussion are outside this formalization.
Operational achievability uses the external private coding theorem
\cite{CWY2004,Devetak2005}.

\begin{table}[H]
\centering
\small
\setlength{\tabcolsep}{4pt}
\renewcommand{\arraystretch}{1.08}
\caption{Paper--Lean correspondence. Unqualified names belong to
\texttt{QIT.Transition}. Definitions are marked \emph{def.}; the
remaining entries are proved statements.}
\label{tab:lean-correspondence}
\vspace{4pt}
\begin{tabular}{@{}>{\raggedright\arraybackslash}p{.42\linewidth}
          >{\raggedright\arraybackslash}p{.55\linewidth}@{}}
\toprule
Paper statement or object & Lean declaration\\
\midrule
\multicolumn{2}{@{}l}{\textit{Channels and zero capacities}}\\[4pt]
Capacity formula~\eqref{eq:capacity}
& \texttt{QIT.Channel.privateCapacity} (def.)\\[5pt]
Main channel~\eqref{eq:main-channel} and its full complement
& \texttt{mainChannel} (def.)\newline
\texttt{complement\_isComplementOf}\\[5pt]
Simulator matrices, Appendix~\ref{app:simulator}, and
$\cD\cN^c=\T\cN$
& \texttt{QIT.TransitionPaperData.simulator\_formula}\newline
\texttt{simulator\_identity}\\[5pt]
$P(\cN)=P(\cE_{2,p})=0$, $1/2\le p\le1$
& \texttt{individual\_capacities\_zero}\\
\midrule
\multicolumn{2}{@{}l}{\textit{Encoding, measurement, and environmental information}}\\[4pt]
Input states~\eqref{eq:inputs} and helper marginal $I_R/2$
& \texttt{inputState} (def.)\newline
\texttt{inputState\_helper\_marginal}\\[5pt]
Actual letters $\rho_0,\rho_t$ and priors $3/4,1/4$
& \texttt{mixedSignal}, \texttt{ensemble} (def.)\\[5pt]
Error detection~\eqref{eq:detection}
& \texttt{QIT.TransitionPaperData.error\_detection}\\[5pt]
Fixed binary POVM $\{F_1,I-F_1\}$
& \texttt{receiverEffect} (def.)\newline
\texttt{receiverEffect\_pos},
\texttt{receiverEffect\_le\_one}\\[5pt]
Bob's information and linear bound~\eqref{eq:bob-gain}
& \texttt{measured\_information}\newline
\texttt{rare\_signal\_gain}\\[5pt]
Full environment~\eqref{eq:environment}, with both erasure branches
& \texttt{product\_isComplementOf}\newline
\texttt{full\_environment}\\[5pt]
Environmental order bounds~\eqref{eq:order}
& \texttt{environment\_order}\newline
\texttt{erased\_environment\_order}\\[5pt]
Classical coin bound, Lemma~\ref{lem:coin}
& \texttt{QIT.DominatedCoin.holevo\_quadratic\_bound}\\[5pt]
Branch-weighted leakage~\eqref{eq:eve-cost}
& \texttt{QIT.FlaggedCoin.holevo\_bound}\newline
\texttt{environment\_information\_upper}\\
\midrule
\multicolumn{2}{@{}l}{\textit{Information advantage and superactivation}}\\[4pt]
Fixed-measurement bound~\eqref{eq:gain} and optimized rate
& \texttt{measuredRate\_lower}\newline
\texttt{certifiedRate\_le\_measuredRate}\\[5pt]
Theorem~\ref{thm:private}: the two zeros and
$P(\cE_{2,p}\otimes\cN)\ge r_p>0$
& \texttt{superactivation\_main}\\[5pt]
Half-erasure specialization, $r_{1/2}=3\ln2/10927$
& \texttt{half\_erasure\_rate}\\
\bottomrule
\end{tabular}
\end{table}

\end{document}